\documentclass{elsart}
\usepackage{epsf}
\usepackage{amssymb,amsmath,amstext,amsfonts}

\def\cal#1{\fam2#1}

\newtheorem{theorem}{\sc Theorem}
\newtheorem{lemma}{\sc Lemma}
\newtheorem{prope}{\sc Property}
\newtheorem{coro}{\sc Corollary}
\newtheorem{nota}{\sc Notation}
\newtheorem{defin}{\sc Definition}
\newenvironment{proof}{\par \sc Proof.\rm}{\hspace*{\fill}$\Box$\vspace{1ex}}
\newenvironment{remark}{\begin{rem}}{\hspace*{\fill}$\Diamond$\end{rem}}
\newenvironment{corollary}{\begin{coro}}{\end{coro}}

\newenvironment{definition}{\begin{defin}}{\end{defin}}

\def\ds{\displaystyle}

\def\nn{\mathbb N}

\def\nn{\mathbb N}

\newcommand{\ld}[1]{\mbox{depth}_{#1}}

\begin{document}
\date{}
\begin{frontmatter}

\title{On the Logical Depth Function and the Running Time of Shortest Programs}

\thanks[title]{This work was supported by FCT projects PEst-OE/EEI/LA0008/2011 and PTDC/EIA-CCO/099951/2008. The authors are also supported by the grants SFRH/BPD/76231/2011 and SFRH/BD/33234/2007 of FCT.\\
Adress authors 1,3: Departamento de Ci\^encia de Computadores
R.Campo Alegre, 1021/1055, 4169 - 007 Porto - Portugal.
Email author 1: {\tt lfa@dcc.fc.up.pt},Email author 3: {\tt andreiasofia@ncc.pt}.\\
Adress author 2: Departamento de Matem\'atica IST
Av. Rovisco Pais, 1, 1049-001 Lisboa - Portugal.
Email: {\tt a.souto@math.ist.utl.pt}.\\
Author 4: CWI, Science Park 123, 1098XG Amsterdam, The Netherlands.
Email: {\tt Paul.Vitanyi@cwi.nl}.
}

\author{L. Antunes$^1$ \and A. Souto$^2$ \and A. Teixeira$^3$ \and 
P.M.B. Vit\'anyi$^4$}

\address{1,3 Instituto de Telecomunica\c{c}\~{o}es 
and\\
 Faculdade de Ci\^{e}ncias Universidade do Porto
\\ 2 Instituto de Telecomunica\c{c}\~{o}es and \\ 
Instituto Superior T\'ecnico, Universidade T\'ecnica de Lisboa
\\4 CWI and University of Amsterdam}


\begin{abstract}
For a finite binary string $x$ its logical
depth $d$ for significance $b$ is the shortest running time of
a program for $x$ of length $K(x)+b$. There is another definition
of logical depth. We give a new proof that the two versions are close.
There is a sequence of strings of consecutive lengths 
such that for every string there is a $b$ such that
incrementing $b$ by 1 makes the associated depths go
from incomputable to computable. The maximal gap between depths
resulting from incrementing appropriate $b$'s by 1 is incomputable. 
The size of this gap is upper bounded by the Busy Beaver function.
Both the upper and the lower bound  hold for
the depth with significance 0.
As a consequence,
the minimal computation time of the associated shortest programs 
rises faster than any computable function 
but not so fast as the Busy Beaver function.

{\bf Classification}:
 Logical depth, Kolmogorov complexity, Information measures, 
Busy Beaver function.
\end{abstract}
\end{frontmatter}

\section{Introduction}

The logical depth is related to complexity with bounded resources.
Computing a string $x$ from one of its shortest programs 
may take a very long time.
However, computing the same string from a program a simple `print$(x)$'
program of length about $|x|$ bits takes very little time.

A program for $x$ of larger length than a given program 
for $x$ may decrease the computation time but except for
pathological cases does not increase it.
Therefore, except for pathological cases we associate
the longest computation time with a shortest
program for $x$. Such a program is incompressible. 
There arises the question how much time can be saved 
by computing a given string from a $b$-incompressible  program 
(a program that can be compressed by at most $b$ bits) when $b$ rises.
\subsection{Related Work}
The minimum time to compute a string by a $b$-incompressible 
program was first considered in
\cite{ben88}. This minimum time
is called the {\em logical depth} at significance 
$b$ of the string concerned.
Definitions, variations, discussion and early results can be found
in the given reference.
A more formal treatment as well as an intuitive approach was given in
the textbook \cite{LiVi}, Section 7.7. 
In \cite{afmv06} the notion of {\em computational}
depth is defined as $K^d(x) -K(x)$ (see definitions below).
This would equal the negative logarithm of the expression
$Q^d_U(x)/Q_U(x)$ in Definition~\ref{def.ld} if the following were 
proved. Since \cite{lev74}
proved in the so-called Coding Theorem that $-\log Q_U(x) = K(x)$ up
to a constant additive term it remains to prove 
$-\log Q_U^d (x) = K^d(x)$ up to a small 
additive term. The last equality
is a major open problem in Kolmogorov complexity theory,
see \cite{LiVi} Exercises 7.6.3 and 7.6.4.

\subsection{Results}
All computations below vary by the choice
of reference optimal prefix machine.
We prove that there is an infinite sequence of strings, 
$x_1,x_2, \ldots$ with $|x_{m+1}|=|x_m|+1$, 
such that for every $m>0$ each $x_m$ is computed by 
$b_1^m,b_2^m$-incompressible programs 
in $d_1^m,d_2^m$ steps, respectively, 
with $b_2^m=b_1^m+1$ 
while $d_1^m-d_2^m$ rises faster than any computable function
but not faster than the Busy Beaver function
the first incomputable function \cite{Ra62} 
(Theorem~\ref{theo.3} and Corollary~\ref{cor.2}). 
We call $b_j^m$ the ``significance'' of ``logical
depth'' $d_j^m$ of string $x_m$ ($j =1,2$).
We prove next (Theorem~\ref{theo.BB}) for this 
infinite sequence of strings that 
the running times of shortest programs associated with
this sequence 
rise faster than any computable function but again not
so fast as the Busy Beaver function. We also give a new proof
(provided by a referee) of a new result 
showing the closeness of the two versions of logical depth
(Theorem~\ref{Teo1}).

The rest of the paper is organized as follows: in Section~\ref{sect.2}, 
we introduce some notation, definitions and basic results needed 
for the paper. 
In Section~\ref{instabilidade}, we prove the results mentioned.

\section{Preliminaries}\label{sect.2}
We use {\em string} or {\em program} to mean a finite binary string. The
alphabet is $\Sigma=\{0,1\}$, and $\Sigma^*= \{0, 1\}^*$ is the 
set of all strings. Strings are denoted by the 
letters $x$, $y$ and $z$. The {\em length} of a 
string $x$ (the number of occurrences of bits in it) is
denoted by $|x|$, and the {\em empty} string by $\epsilon$. Thus,
$|\epsilon|=0$.
We use the notation $\Sigma^n$ for the set of strings of 
length $n$.
We use the binary logarithm which is denoted by ``$\log$.''
Given two functions $f$ and $g$, we say that $f \in O(g)$ if there is a constant $c >0$, 
such that $f(n) \leq c \cdot g(n)$, 
for all but finitely many natural numbers $n$. 

\subsection{Time Bounds}
Often the resource-bounds are 
{\em time constructible}. There are many definitions. For example,
there is a Turing machine whose running time is exactly $t(n)$ 
on every input of size $n$, for some function $t: \nn \to \nn$,
where $\nn$ is the set of natural numbers. However, 
in this paper there are functions that 
are not time constructible. An example
is the Busy Beaver function $BB: \nn \to \nn$ (Definition~\ref{bbfunction}) 
which is not computable
(it rises faster than any computable function). 
In this case, and when we just mean a number of steps, we 
indicate in the superscript the number of steps taken, usually
using $d$.
 Given a program $p$, we denote its running time (the number of
steps taken by the reference optimal prefix Turing machine defined
below) by $time(p)$. 

\subsection{Computability}
A pair of nonnegative integers,
such as $(p,q)$ can be interpreted as the rational $p/q$.
We assume the notion of a computable function with rational arguments
and values.
A function $f(x)$ with $x$ rational is \emph{semicomputable from below}
if it is defined by a rational-valued total computable function $\phi(x,k)$
 with $x$ a rational number
and $k$ a nonnegative integer
such that $\phi(x,k+1) \geq \phi(x,k)$ for every $k$ and
  $\lim_{k \rightarrow \infty} \phi (x,k)=f(x)$.
This means
that $f$ (with possibly real values)
can be computed in the limit from below
(see \cite{LiVi}, p. 35).  A function $f$ is  \emph{semicomputable
from above} if $-f$ is semicomputable from below.
 If a function is both semicomputable from below
and semicomputable from above then it is \emph{computable}.

\subsection{Kolmogorov Complexity}

We refer the reader to the textbook
\cite{LiVi} for details, notions, and history.
We use Turing machines with a read-only one-way input tape, one or more
(a finite number) of work tapes at which the computation takes place,
and a one-way write-only output tape. All tapes are semi-infinite,
divided into squares, and each square can contain a symbol 
from a given alphabet or blanks.
The machine uses for all of its
tapes a finite alphabet and all tapes are one-way infinite. 
Initially, the input tape is inscribed with a 
semi-infinite sequence of 0's and 1's.
The other tapes are empty (contain only blanks).
At the start, all tape heads scan 
the leftmost square on their tape. If the machine
halts for a certain input then the contents of the scanned segment
of input tape is called the {\em program} or {\em input}, and the contents of
the output tape is called the {\em output}. The machine thus described
is a {\em prefix Turing machine}. Denote it by $T$.
If $T$ terminates, then the program is $p$
and the output is $T(p)$. The set 
${\cal P}=\{p:T(p)<\infty \}$ is {\em prefix-free} (no element of the set
is a proper prefix of another element). By the ubiquitous Kraft inequality
\cite{Kr49} we have 
\begin{equation}\label{eq.kraft}
\sum_{p \in {\cal P}} 2^{-|p|} \leq 1.   
\end{equation}
The same holds for a fixed conditional or auxiliary. 
The above unconditional case 
corresponds to the case where the conditional is $\epsilon$.
Among the universal prefix-free Turing machines we consider a special
subclass called {\em optimal}, see Definition 2.0.1 in \cite{LiVi}.
To illustrate this concept: let $T_1,T_2, \ldots$ be a standard enumeration of (prefix)
Turing machines, and let $U_1$ be one of them.  If $U_1 (i,pp) = T_i(p)$
for every index $i$ and program $p$ and outputs 0 for inputs that are
not of the form $pp$ (doubling of $p$), 
then $U_1$ is also universal. However,
$U_1$ can not be used to define Kolmogorov complexity. For that we need 
a machine $U_2$ such that $U_2(i,p)=T_i(p)$ for every $i,p$. 
A machine such as $U_2$ is called an {\em optimal} prefix Turing machine.
Optimal prefix Turing machines are a strict subclass of 
universal prefix Turing machines. The above
example illustrates the strictness. To define Kolmogorov
complexity we require optimal prefix Turing machines
and not just universal prefix Turing machines. 
The term `optimal' comes from the founding paper \cite{Ko65}.

It is still possible that two different optimal prefix Turing machines 
have different computation times for the same input-output pairs or
even different sets of programs. To avoid these problems we
fix a reference machine. Necessarily, the reference machine has a 
certain number of worktapes.
A well-known result of \cite{HeSt66} states that $n$ steps of 
a multiworktape
prefix Turing machine can be simulated in $O(n \log n)$ steps
of a two-worktape prefix Turing machine. Thus,
for  such a simulating optimal Turing machine $U'$ we have
$U'(i,p)=T_i(p)$ for all $i,p$; if $T_i(p)$ terminates in time $t(n)$
then $U'(i,p)$ terminates in time $O(t(n)\log t(n))$. 
Altogether, we fix such an simulating (as above)
optimal prefix Turing machine $U'$ and call it the {\em reference
optimal prefix Turing machine} $U$.

\begin{definition}
Let $U$ be the reference optimal prefix-free Turing machine, 
and $x,y$ be strings. 
The \em prefix-free Kolmogorov complexity \em 
$K(x|y)$ of $x$ given $y$ is defined by 
$$
  K(x|y)=\min \{|p| : U(p,y)= x \}.
$$
The notation $U^d(p,y)=x$ means that $U(p,y)=x$ within
$d$ steps.
The \em $d$-time-bounded
prefix-free Kolmogorov complexity \em $K^{d}(x|y)$ of $x$ given $y$ is
defined by
$$K^{d}(x|y) =\min \{|p| : U^d(p,y)=x \}.$$
\end{definition}

The default value for the auxiliary input $y$ for the program $p$, 
is the empty string $\epsilon$. To avoid overloaded notation we usually 
drop this argument in case it is there. 
Let $x$ be a string. Denote by $x^*$ the first shortest program in standard
enumeration such that $U(x^*)=x$.

\begin{definition}
Let $|x|=n$.
The string $x$ is \em $c$-incompressible \em if \[K(x)\geq n+K(n)-c.\]
The string $x$ given $n$ is \em $c$-incompressible \em if \[K(x|n)\geq n-c.\]
\end{definition}

A simple counting argument can show the existence of $c$-incompressible strings
($c$-incompressible strings given their length)
of every length for the plain complexity $C(x)$. 
Since $K(x) \geq C(x)$ ($K(x|n \geq C(x|n)$ we have the following:
\begin{lemma}\label{lemmai}
There are at least $2^n(1-2^{-c})+1$ strings $x\in \Sigma^n$ (given $n$)
that are $c$-incompressible with respect to prefix Kolmogorov complexity.
\end{lemma}

\section{Logical Depth}
The logical depth \cite{ben88} consists of two versions.
One version is based on $Q_U(x)$, 
the so-called {\em a priori} probability \cite{LiVi} and its 
time-bounded version. Here $U^d(p)$ means that $U(p)$ terminates in
at most $d$ steps. 
\[
Q_U(x)=\sum_{U(p)=x}2^{-|p|}, \;\;\;\;\;Q^d_U(x)=\sum_{U^d(p)=x}2^{-|p|}
\]
For convenience we drop the subscript and consider $U$ as understood.
\begin{definition}\label{def.ld}
Let $x$ be a string, $b$ a nonnegative integer.
The logical depth, tentative version 1, 
of $x$ at significance level $\varepsilon= 2^{-b}$ is
\[
\ld{\varepsilon}^{(1)}(x) = \min \left\{d:\ds\frac{Q^d(x)}{Q(x)}\geq \varepsilon\right\}
\]
\end{definition}
Using a program that is longer than another program for output $x$
can shorten the computation time. Thus,
the $b$-significant logical depth of an object $x$ is defined 
as the minimal time 
the reference optimal prefix Turing machine needs 
to compute $x$ by a $b$-incompressibel
program (one that can be compressed by at most $b$ bits).
\begin{definition}\label{logicaldepth}\label{def.final}
Let $x$ be a string, $b$ a nonnegative integer. 
The logical depth, tentative version 2, of $x$ at significance level 
$b$, is: 
\[
\ld{b}^{(2)}(x) = \min \{time(p) : |p| \leq K(p) + b \wedge U(p) = x\} \,.
\] 
In this case we say that the string $x$ is {\em $(d,b)$-deep}.
\end{definition}
\begin{remark}
\rm
It is easy to see that $\ld{0}^{(2)}(x)$ is the least number of steps to compute $x$ 
from an incompressible program. For example, $x^*$ is known to be 
incompressible up to an additive constant \cite{LiVi}. 
Thus, $time(x^*) \geq \ld{0}^{(2)}(x)$. 
For higher $b$ the value of $\ld{b}^{(2)}(x)$
is monotonic nonincreasing until 
$$\ld{|x|-K(x)+O(1)}^{(2)}(x)=O(|x| \log |x|),$$
the $O(1)$ term represents a program to copy the literal representation
of $x$ in $O(|x| \log |x|)$ steps. 
It is the aim of this paper to study the 
properties the graph of $f_x$ can have. For example, if $x$ is random
(i.e., $|x|=n$ and $K(x) \geq n+K(n)$) then always $b=O(1)$ and always $d=O(n \log n)$.
These $x$'s, but not only these, are called {\em shallow}.
\end{remark}

Version (2) is stronger than version (1) in that in the version (2)
every individual
program at significance level $b$ must take at most $\ld{b}^{(2)}(x)$ steps
to compute $x$, while version (1) to be equivalent would 
require only that a weighted average of
all programs for $x$ require at most $\ld{2^{-b}}^{(1)}(x)$ steps. 
The quantitative difference between the two versions of logical
depth is small as the
following statement shows. Compare with
Theorem 7.7.1 together with Exercise 7.7.1 in \cite{LiVi}.
The anonymous referee provided a slightly different
statement and the new proof below.
\begin{theorem}\label{Teo1}
Let $x$ be $(d,b)$ deep. Then
\[
\frac{1}{2^{b+K(b)+O(1)}} \leq \frac{Q^d(x)}{Q(x)}
< \frac{1}{2^{b+1}}.
\]
\end{theorem}
\begin{proof}
(Right $\leq$) This follows from: {\em If for integers $d,b>0$ the total
a priori probability $Q^d(x)$ of all programs that compute $x$ 
within $d$ steps
is at least $2^{-b-1} Q(x)$, then one of these programs is 
$b+1$-incompressible}. Indeed, if for some $c$
all programs computing $x$ within $d$ steps are $c$-compressible, 
then the twice iterated
reference optimal Turing machine (in its role as decompressor) computes
$x$ with probability $2^c Q^d(x) \geq 2^{c-b-1}Q(x)$ 
from the $c$-compressed versions.
But $Q(x) \geq 2^{c-b-1}Q(x)+2^{-b-1} Q(x)$.
Hence $c-b-1 < 0$  that is $c < b+1$. Hence there is a program computing
$x$ within $d$ steps that is $b+1$-incompressible. Then
$x$ may be $(d,b+1)$-deep contradicting the assumption that
$x$ is $(d,b)$-deep.
Hence $Q^d(x) < 2^{-b-1} Q(x)$.

(Left $\leq$) This follows from: {\em If for integers $d,b$ there exists a 
$b$-incompressible program that computes $x$ in time $d$, then 
$Q^d(x) \geq 2^{-b-K(b)-O(1)} Q(x)$}. 
Assume by way of contradiction 
that $Q^d(x) < 2^{-B} Q(x)$ for some $B$ (see below about the choice
of $B$). 
Consider the following lower semicomputable
semiprobability (the total probability is less than 1):
for each string $x$ we enumerate
all programs $p$ that compute $x$ in order of halting (time), and
assign to each halting $p$ the probability $2^{-|p|+B}$ until
the total probability would pass $Q(x)$ with the next halting $p$.
Since $Q(x)$ is lower semicomputable we can postpone 
assigning probabilities.
But eventually or never for some $p$
the total probability may pass $Q(x)$ and this $p$ and all subsequent
halting $p$'s for $x$ get assigned probability 0.   
Therefore, the total probability assigned to all halting programs
for $x$ is less than $Q(x)$. Since by assumption $Q^d(x) < 2^{-B} Q(x)$
we have
$2^B Q^d(x) = \sum_{U^d(p)=x} 2^{-|p|+B} < Q(x)$. Since $Q(x) < 1$ we
have $Q^d(x) < 2^{-B}$ and therefore all programs that
compute $x$ in at most $d$ steps are $(B-O(1))$-compressible
given $B$, and therefore $(B-K(B)-O(1))$-compressible.

By assumption there exists a $b$-incompressible 
program from which $x$ can be computed in $d$ steps. We obtain
a contradiction with this for $B-K(B)-O(1)  \geq b$. This is the case if 
we set $B = b+K(b)+c$ for a large enough constant $c>0$.
Namely $B-K(B)-O(1)=
b+K(b)+c-K(b+K(b)+c)-O(1) \geq b+K(b)+c-K(b+K(b))-K(c)-O(1)>b$
(note that $K(b+K(b)) \leq K(b,K(b))+O(1) \leq K(b)+O(1)$
by some easy argument \cite{LiVi} and $K(c) = O(\log c) < c/2$
for large $c$). Therefore, $Q^d(x) < 2^{-B} Q(x)$ is contradicted.
Hence $Q^d(x) \geq 2^{-B} Q(x)$ which implies
$Q^d(x) \geq 2^{-b-K(b)-O(1)} Q(x)$. 
By assumption $x$ is $(d,b)$-deep satisfying the condition of 
this case.
\end{proof}

\begin{remark}
\rm
We can replace $K(b)$ by $K(d)$ by changing the constriction of the
semiprobability: knowing $d$ we generate all programs that compute 
$x$ within $d$ steps and let the semiprobabilities be propertional to 
$2^{-|p|}$ and the sum be at most $Q(x)$. In this way $K(b)$ 
in Theorem~\ref{Teo1} is substituted by $\min\{K(b),K(d)\}$. 
\end{remark}
Theorem~\ref{Teo1} shows that the quantitative difference between
the two versions of logical depth
are small. 
We choose version 2 as our final definition of logical depth.


\begin{definition}\label{bbfunction}
The \em Busy Beaver function \em $BB: \nn \rightarrow \nn$ is defined by
\[
BB(n) = \max_{p:|p|\leq n}\{\mbox{\rm running time of $U(p) < \infty$}\}
\]
\end{definition}

\section{The graph of logical depth}\label{instabilidade}

Even slight changes of the 
significance level $b$ can cause large changes
in logical depth.

\begin{lemma}\label{lemma1}
Let $n$ be large enough. There exist strings $x$ of length $n$ 
such that the running time of a computation from a shortest program
to $x$ is incomputable.
\end{lemma}
\begin{proof}
By \cite{Ga74} we have $K(K(x)|x) \geq \log n - 2 \log \log n-O(1)$.
(This was improved to the optimal $K(K(x)|x) \geq \log n -O(1)$ recently
in \cite{BS03}.) Hence there is no computable function $\phi(x) =
\min\{d: U^d(p)=x ,\; |p|=K(x)\}$. If there were, then we could run
$U$ for $d$ steps on any program of length $n+O(\log n)$. Among the 
programs which halt within $d$ steps we select the ones which output $x$.
Subsequently, we select from this set a program of minimum length.
This is a shortest program for $x$ of length $K(x)$. Therefore,
the assumption that $\phi$ is computable implies that $K(K(x)|x = O(1)$:
contradiction.  
\end{proof}
The following result was mentioned informally in \cite{ben88}.
\begin{lemma}\label{lemma2}
The running time of a program $p$ 
is at most $BB(|p|+O(\log |p|))$. 
The running time of a shortest program for a string $x$ of length $n$
is at most $BB(n+O(\log n))$.
\end{lemma}
\begin{proof}
The first statement of the lemma follows from 
Definition~\ref{bbfunction}. For the second statement we use
the notion of a simple prefix-code called a {\em self-delimiting} code. 
This is obtained by reserving one
symbol, say 0, as a stop sign and encoding a string $x$ as $1^x 0$.
We can prefix an object with its length and iterate
this idea to obtain ever shorter codes:
$\bar{x}= 1^{|x|} 0 x$ with
length $|\bar{x}| = 2|x| + 1$, and 
$x'= \overline{|x|} x$ of length $|x|+2||x||+1=
|x|+O(\log |x|)$ bits. From this code $x$ is readily extracted.
The second statement follows since $K(x) \leq |x|+O(\log |x|)$. 
\end{proof}

\begin{theorem}\label{theo.3}
Let $n$ be large enough.
There is a string $x$ of every length $n$ such 
that $\ld{i}(x)$ is incomputable and
$\ld{i+1}(x)$ is computable for some $i$ 
that satisfies $0 \leq i \leq n+O(\log n)$.
\end{theorem}

\begin{proof}
Let $x$ be a string of length $n$ as in Lemma~\ref{lemma1}. 
Then $\ld{0}(x)$ is incomputable.
However, $\ld{n-K(x)+O(\log n)}(x)= O(n \log n)$ and therefore
computable. Namely,
a self-delimiting encoding
of $x$ can be done in $n+O(\log n)$ bits. Let $q$ be such an encoding
with $q=1^{||x||}0|x|x$ (where $||x||$ is the length of $|x|$).
Let $r$ be a self-delimiting program of $O(1)$ bits which prints the
encoded string in the next self-delimiting string. Consider the program $rq$.
Since $x$ can be compressed to length $K(x)$, 
the running time  $\ld{n-K(x)+O(\log n)}(x)$ is at most
the running time of $rq$ which is $O(n \log n)$.

Consider the sequence $\ld{0}(x), \ldots , \ld{n-K(x)+O(\log n)}(x)$.
Since $\ld{0}(x)$ is incomputable and $\ld{n-K(x)+O(\log n)}(x)$ is
computable, there must be an $i$ satisfying $0 \leq i < n-K(x)+O(\log n)$
such that $\ld{i}(x)$ is incomputable and
$\ld{i+1}(x)$ is computable. Let set of all such $i$'s is finite
and discrete.
Hence it has a maximum.
\end{proof}

Let $x$ be as in Theorem~\ref{theo.3} and
$i_{\max}$ be the $i$ which reaches the maximum 
of $\ld{i}(x)-\ld{i+1}(x)$, 
$0 \leq i \leq n+O(\log n)$.

\begin{corollary}\label{cor.2}
\rm
There exists an infinite sequence of strings $x_1, x_2, \ldots$ with 
$|x_{m+1}|=|x_m|+1$, each $x_m$ with an $i_{m,\max}$ ($0 \leq i < |x_m|-K(x_m)
+O(\log |x_m|)$), and $h(m)=\ld{i_{m, \max}}(x_m)-\ld{i_{m+1, \max}}(x_m)$ is
the maximal gap in the logical depths of which 
the significance differs by 1.
The function $h(m)$ rises faster than any computable function but is not
faster than $BB(|x_m|+O(\log |x_m|)$ by Lemma~\ref{lemma2}.
\end{corollary}
\begin{theorem}\label{theo.BB}
There is an infinite sequence $x_1,x_2, \ldots$
with $|x_{m+1}| = |x_m|+1$  such that 
$h(m) \leq \ld{0}(x_m)  \leq BB(|x_m|+O(\log |x_m|))$.
(Note that $\ld{0}(x_m)$ is the shortest time of a computation 
of a shortest program for $x$ by $U$.)
\end{theorem}
\begin{proof}
\rm
The logical depth function $\ld{b}(x)$ 
is monotonic nonincreasing in the significance argument $b$
for all strings $x$ by its Definition~\ref{def.final}. 
By Lemma~\ref{lemma2} and Corrollary~\ref{cor.2} the theorem follows.
\end{proof}

\section{Conclusion}

We studied the behavior of the 
logical depth function associated
with a string $x$ of length $n$. 
This function is monotonic nonincreasing. For
argument 0 the logical depth is the minimum running time of the
computation from a shortest program for $x$ to $x$. The
function decreases to $O(n \log n)$ for the argument 
$|x|-K(x)+O(\log |x|)$.
We show that there is an infinite sequence of strings such that 
the difference in 
logical depths of some significance levels differing by one
rises faster than any computable function,
that is, incomputably fast, but not more than the Busy Beaver function.
This shows that 
logical depth can increase tremendously for
only an incremental difference in significance.
Moreover, there is an infinite sequence of strings such that the
minimal computation times of associated shortest 
programs 
rises incomputably fast  but not so fast as the 
Busy Beaver function with as argument
twice the length of the string plus an 
additive logarithmic term in the length.

\section*{Acknowledgments}
We thanks Bruno Bauwens for helpful discussions and comments, 
and the anonymous referee for additional comments and the new proof
of Theorem~\ref{Teo1}.

\bibliographystyle{plain}

\end{document}